\documentclass[a4paper,10pt]{article}

\usepackage{amsmath}
\usepackage{amssymb}
\usepackage{amsthm}
\usepackage{amsfonts}
\usepackage{xcolor}
\usepackage{graphicx}
\usepackage{mathtools}
\usepackage{enumerate}
\usepackage{multirow}
\usepackage{rotating}
\usepackage{setspace}
\usepackage{lineno}
\usepackage{url}

\theoremstyle{theorem}
\newtheorem{thm}{Theorem}

\theoremstyle{definition}

\newcommand\num{\addtocounter{equation}{1}\tag{\theequation}}

\usepackage{geometry}

\begin{document}
\title{Optimal Drug Regimen and  Combined Drug Therapy and its Efficacy in the Treatment of COVID-19 :  An Within-Host Modeling Study}
	
\vspace{0.1in}
\author{{\bf\large Bishal Chhetri$^{1,a}$, Vijay M. Bhagat$^{1,b}$, D. K. K. Vamsi$^{*,a}$, Ananth V S$^{a}$,  Bhanu Prakash$^{a}$}, \\ {\bf\large Swapna Muthusamy$^{b},$ Pradeep Deshmukh$^{c},$ Carani B Sanjeevi$^{d, e}$}\hspace{2mm} \\
	{\it\small $^{a}$Department of Mathematics and Computer Science, Sri Sathya Sai Institute of Higher Learning, Prasanthi Nilayam}, \\
	{\it \small Puttaparthi, Anantapur District - 515134, Andhra Pradesh, India}\\
	{\it\small $^{b}$Central Leprosy Teaching and Research Institute - CLTRI, Chennai, India}\\
	{\it\small $^{c}$ Professor and Head, Department of Community Medicine, All India Institute of Medical Sciences - AIIMS, Nagpur,India}\\
	{\it\small $^{d}$  Vice-Chancellor, Sri Sathya Sai Institute of Higher Learning -  SSSIHL, India}\\
	{\it\small $^{e}$ Department of Medicine, Karolinska Institute, Stockholm, Sweden }\\
	{\it\small bishalchhetri@sssihl.edu.in,    vijaydr100@gmail.com, dkkvamsi@sssihl.edu.in$^{*}$,}\\
	{\it\small ananthvs@sssihl.edu.in, prakashdmacs@gmail.com,      swapnamuthuswamy@gmail.com,}\\
	{\it\small	prdeshmukh@gmail.com, sanjeevi.carani@sssihl.edu.in, sanjeevi.carani@ki.se}\\
	{\small $^{1}$ Joint First Authors},
	{ \small $^{*}$ Corresponding Author}
	\vspace{1mm}
}
	\date{}
	\maketitle
	
	\begin{abstract} \vspace{.25cm}
	
	The COVID-19 pandemic has resulted in more than 30.35 million infections and 9, 50, 625 deaths in 212 countries over the last few months. Different drug intervention acting at multiple stages of pathogenesis of COVID-19 can substantially reduce the infection induced mortality. The current within-host mathematical modeling studies deals with the optimal drug regimen and the efficacy of combined therapy in treatment of COVID-19. The drugs/interventions considered include Arbidol, Remdesivir, Inteferon (INF) and Lopinavir/Ritonavir. It is concluded that these drug interventions when administered individually or in combination  reduce the infected cells and viral load. Four scenarios involving administration of single drug intervention, two drug interventions, three drug interventions and all the four  have been discussed.  In all these scenarios the optimal drug regimen is proposed  based on two methods. In the first method these medical  interventions are modeled as control interventions and a corresponding objective function and optimal control problem is formulated. In this setting the optimal drug regimen is proposed. Later using the the comparative effectiveness method the optimal drug regimen is proposed based on basic reproduction number and viral load.  The  average  infected cell count and  viral load decreased the most when all the four interventions were applied together. On the other hand the average susceptible cell count decreased the best when   Arbidol alone was administered.  The basic reproduction number and viral count decreased the best when all the four interventions were applied together reinstating the fact obtained earlier in the optimal control setting. These findings may help physicians with decision making in treatment of life-threatening COVID-19 pneumonia.
	
	\end{abstract}
	
	{\bf{Keywords}}: COVID-19; Inflammatory Mediators;   Drug Interventions; Arbidol; Remdesivir; Interferon; 
	Lopinavir/Ritonavir; Optimal Control Problem; Optimal Drug Regimen; Comparative Effectiveness Study
	\vspace{.12 in}
	
	\section{Introduction} \vspace{.25cm}
	
	\quad The unprecedented pandemic affecting more than 212 countries by infection of Severe Acute Respiratory Syndrome Coronavirus-2 (SARS-CoV-2) causing Covid-19.  There are already 30.35 million  infected individuals and 9, 50, 625 deaths worldwide, and further daily exponential increase \cite{1}. \\
	
	\quad A detailed pathogenesis and a corresponding within-host model involving  crucial inflammatory mediators and the host immune response has been developed and discussed at length by the authors elsewhere  \cite{ref1}. This work dealt with the natural history and course of infection of COVID-19. The authors also briefly discussed about the optimality and effectiveness  of combined therapy involving one or more antiviral and one or more immuno-modulating drugs when  administered together. \\  
	
	\quad  The present work deals with optimal drug regimen and the efficacy of combined therapy in treatment of COVID-19.  The drug interventions considered include Arbidol, Remdesivir, Interferon and Lopinavir/Ritonavir. It is concluded that these drug when administered individually or in combination reduce the  infected cells and viral load significantly. Four scenarios involving administration of single drug, two drugs, three drugs and all the four drugs have been discussed.  In all these scenarios the optimal drug regimen is proposed based on two methods. The first method in the optimal control problem setting and the second method being comparative effectiveness studies.  The average  infected cell count and  viral load decreased the most when all the four interventions were applied together. On the other hand, the average susceptible cell count decreased the best when   Arbidol alone was administered. The comparative effectiveness studies showed that the  basic reproduction number and virus count decreased the best when all the four interventions were applied together reinstating the fact obtained earlier in the optimal control setting. \\
	
	\quad  Some of the important mathematical modelling studies that deal with transmission and spread of COVID-19 at the population level can be found in  \cite{5, 7, 4, 6}. In the recent work  \cite{ea2020host}, an in-host modelling study deals with the qualitative characteristics and estimation of standard parameters of corona viral infections. The optimal control studies involving four clinical drug interventions which is being attempted here is the first of its kind for COVID-19. \\
	
	 {\flushleft \bf{OBJECTIVES }} \vspace{.25cm}
	\begin{itemize}
		\item [1.] 	To investigate the role of pharmaceutical interventions such as  Arbidol, Remdesivir, Interferon and Lopinavir/Ritonavir by incorporating them  as controls at specific sites in the pathogenesis.
		\item [2.] 	To study and compare the dynamics of susceptible, infected cells and viral load with and without these control interventions by studying them as optimal control problems. 
		\item[3.] To propose the optimal drug regimen in four scenarios involving administration of single drug, two drugs, three drugs and all the four drugs  based on the average susceptible, infected cell count, average viral load and basic reproduction number.
		\item [4.] To propose the optimal drug regimen using the comparative effectiveness studies.
	\end{itemize} \vspace{.25cm}

	\begin{minipage}{\textwidth}

		{\bf{\textit{Research In Context}}} \\ 
		
		{\bf{\textit{Evidence before the study:}}}  The efficacy of drugs in the treatment of COVID-19 infection was assessed and reported in a few clinical trials and observational studies. The role of combination of anti-microbial agents in achieving the favorable outcome is very well established in several viral diseases including HIV, dengue hemorrhagic fever etc. To determine the utility of combined anti-microbials through the clinical studies will take a long way and considering the current mortalities due to COVID-19 cannot be awaited for.    \\

		{\bf{\textit{Added value of this study:}}}  The current research study is the first hand evidence for the efficacy of combined anti-microbial agents for the treatment of COVID-19. The study assessed the combined efficacy of four therapeutic interventions namely Arbidol, Remdesivir, Interferons and Lopinavir/Ritonavir. Optimal reduction in the viral load and the number of infected cells can be significantly achieved when the above four interventions used in combination as compared with the efficacy of individual therapeutic agent.  \\
		
		{\bf{\textit{Implications of all the available evidence: }}}  The mathematical evidence generated out of this study will pave a way for developing a standardized clinical treatment regimen both on therapeutic and research avenues. \\
		
	\end{minipage}
	\vspace{.25cm} 
	
	{\flushleft \bf{METHODS}} \\
	
	A within-host mathematical model incorporating  four drug interventions was studied as an   optimal control problem. The {\emph{Fillipov Existence Theorem}} was used to obtain an optimal solution \cite{liberzon2011calculus}. Characterization of optimal controls is done using Maximum Principle \cite{liberzon2011calculus}.  Numerical simulations  involving with and without control interventions  were performed for obtaining optimal drug regimen.  Comparative effectiveness  method \cite{garira2020development} was used to study the change in basic reproduction number and virus count  for  single and multiple interventions.   \vspace{.25cm}
	
	\section{Optimal Control Studies}  \vspace{.25cm}
	
	\subsection{Optimal Control Problem Formulation} \vspace{.25cm}
	
	\quad Drugs such as {{Remdesivir}} inhibit RNA-dependent RNA polymerase and drugs  {{ Lopinavir/Ritonavir}}  inhibit the viral protease there by reducing  viral replication \cite{tu2020review}. Interferons are broad spectrum  antivirals, exhibiting both direct inhibitory effect on viral replication and supporting an immune response to clear virus infection \cite{wang2019global}. On the other hand drugs such as Arbidol not only inhibits the viral replication but also blocks the virus replication by inhibiting the fusion of  lipid membranes with host cells \cite{ref2}.  \vspace{.25cm}
	
	\quad Motivated by the above clinical findings in similar lines to control problem in \cite{ref1}, we consider a  control problem with the  drug interventions Arbidol, Remdesivir, Lopinavir/Ritonavir and Interferon  as controls.
	
	\begin{eqnarray}
		\frac{dS}{dt}& =&  \omega \ - \beta SV  -\mu_{1A}(t)S -\mu S  \label{sec2equ1} \\ 
		\frac{dI}{dt} &=& \beta SV \ -  { \bigg(d_{1}  + d_{2}  +  d_{4} + d_{5}+ d_{6}\bigg)I} \nonumber \\ &-&\bigg(\mu_{2Rem}(t)+\mu_{2INF}(t)+\mu_{2A}(t)+\mu_{2Lop/Rit}(t)\bigg)I-\mu I   \label{sec2equ2}\\ 
		\frac{dV}{dt} &=& \bigg( \alpha-(\mu_{3Rem}(t)+\mu_{3INF}(t)+\mu_{3A}(t)+\mu_{3Lop/Rit}(t))\bigg) I   \nonumber \\
		&-&  \bigg( b_{1}  + b_{2}   +  b_{4}+ b_{5} + b_{6}\bigg)V    \ -  \mu_{1} V \label{sec2equ3}
	\end{eqnarray}  \\
	
	Let $U_1 = (\mu_{1A}, \mu_{2A}, \  \mu_{3A}), U_2 = (\mu_{2Rem}, \mu_{3Rem}), \  U_3 = (\mu_{2INF}, \mu_{3INF}), \ U_4 = (\mu_{2Lop/Rit}, \mu_{3Lop/Rit}). $ 
	
	\begin{table}[ht!]
		\centering 
		\begin{tabular}{|l|l|} 
			\hline
			\textbf{Parameters} &  \textbf{Biological Meaning} \\  
			\hline 
			$S$ & Healthy Type II Pneumocytes  \\
			\hline
			$I$ & Infected Type II Pneumocytes  \\
			\hline
			$\omega$ & Natural birth rate of Type II Pneumocytes \\
			\hline
			$V$ & Viral load  \\
			\hline
			$\beta$ & Rate at which healthy Pneumocytes are infected  \\
			\hline
			$\alpha$ & Burst rate of virus particles \\
			\hline
			$\mu$ & Natural death rate of Type II Pneumocytes \\
			\hline
			$\mu_{1}$ & Natural death rate of virus \\
			\hline
			$d_{1}, \hspace{.25cm} d_{2}, \hspace{.25cm} d_{4}, \hspace{.25cm} d_{5}, \hspace{.25cm} d_{6}$ & Rates at which Infected Pneumocytes are removed because\\
			&of the release of cytokines and chemokines  IL-6\\
			&  TNF-$\alpha$,  \hspace{.2cm}CCL5, \hspace{.2cm}CXCL8 , \hspace{.2cm}CXCL10   \hspace{.2cm} respectively   \\ 
			\hline
			$b_{1}, \hspace{.25cm} b_{2}, \hspace{.25cm} b_{4}, \hspace{.25cm} b_{5}, \hspace{.25cm} b_{6}$ & Rates at which Virus is removed because of\\
			& the release of cytokines and chemokines  IL-6\\
			&  TNF-$\alpha$,   \hspace{.2cm}CCL5, \hspace{.2cm}CXCL8 , \hspace{.2cm}CXCL10   \hspace{.2cm} respectively   \\ 
			\hline
			$	u_{1A}, \hspace{.25cm} u_{2A}, \hspace{.25cm} u_{3A} $ & Rates at which susceptible cells,\\
			&infected cells and viral load are reduced\\
			& due to Arbidol drug.\\
			\hline
			$u_{2Rem}, \hspace{.25cm} u_{3Rem}$ & Rates at which infected cells and the \\
			& viral load are reduced due to Remdesivir\\
			\hline
			$u_{2INF}, \hspace{.25cm} u_{3INF}$ & Rates at which infected cells and the \\
			& viral load are reduced due to Interferon respectively\\
			\hline
			$u_{2Lop/Rit}, \hspace{.25cm} u_{3Lop/Rit}$ & Rates at which infected cells and the \\
			&viral load are reduced due to Lopinavir/Ritonavir respectively\\
			\hline
		\end{tabular}
	\end{table} \vspace{.25cm}	
	
	A detailed natural history including positivity and boundedness for the system (\ref{sec2equ1}) - (\ref{sec2equ3})  without control interventions is done by the author's team in \cite{ref1}.

	\textbf{Definition of the objective function}
	
	\begin{equation}
		J(U_1, U_2, U_3, U_4 ) = \int_{0}^{T} \bigg(A_1 U_1^2 + A_2 U_2^2 + A_3 U_3^2 + A_4 U_4^2 - I(t) - V(t) \bigg) dt \label{temp}
	\end{equation} 
	
	To reduce the complexity of the problem here we choose to model the control efforts via a linear combination of the quadratic terms. Also when the objective function is quadratic with respect to the control, diﬀerential equations arising from optimization have a known solution. Other functional forms sometimes lead to systems of diﬀerential equations that are diﬃcult to solve \cite{djidjou2020optimal}, \cite{lee2010optimal}.
	
	The integrand of the cost function (\ref{temp}), denoted by
	\begin{displaymath}
		L(S,I,V, U_1, U_2, U_3, U_4 ) = 
		\bigg(A_1 U_1^2 + A_2 U_2^2 + A_3 U_3^2 + A_4 U_4^2 - I(t) - V(t)\bigg)\end{displaymath}
	is called the Lagrangian of the running cost.\\
	
	Here the cost function (\ref{temp}) represents the benefits of each of the interventions and the number of infected cells and viral load throughout the observation period. Our goal is to maximize the benefits of each of the interventions and minimize the infected cell and virus population.\\
	
	The coefficients $A_{i}$, for $i = 1, 2, 3, 4$  are the constants related to the benefits of each of the interventions and their values will be chosen based on the value of the hazard ratio of individual drugs.\\
	
	The admissible solution set for the optimal control problem  (\ref{sec2equ1}) - (\ref{temp}) is given by,
	\begin{displaymath}
		\Omega = \left\{ (S, I, V, U_1, U_2, U_3, U_4)\; | \; S,  I \ and \ V  \text{that satisfy} \ (\ref{sec2equ1})-(\ref{sec2equ3}) \ \forall \ U_i \in U \right\} \end{displaymath}
	
	where $U$ is the set of all admissible controls given by \\
	
	$U = \{U_1 = (\mu_{1A}(t), \mu_{2A}(t), \mu_{3A}(t)), \  U_2 = (\mu_{2Rem}(t), \mu_{3Rem}(t)), \ U_3 = (\mu_{2INF}(t), \mu_{3INF}(t)), \ U_4 = (\mu_{2Lop/Rit}(t), \mu_{3Lop/Rit}(t))\}$:
	
	$	\mu_{1A}(t)) \in [0,\mu_{1A} max], \ \mu_{2A}(t) \in [0,\mu_{2A} max], \ \mu_{3A}(t) \in [0,\mu_{3A} max],  \ \mu_{2Rem}(t) \in [0,\mu_{2Rem} max], \ \mu_{2INF}(t) \in [0,\mu_{2INF} max], \ \mu_{2Lop/Rit}(t) \in [0,\mu_{2Lop/Rit} max], \ \mu_{3Rem}(t) \in [0,\mu_{3Rem} max], \ \mu_{3INF}(t) \in [0,\mu_{3INF} max], \ \mu_{3Lop/Rit}(t) \in [0,\mu_{3Lop/Rit} max], \ t \in [0,T]\}.$
	
	\vspace{.25cm}
	
	\section{Existence of Optimal Controls} \vspace{.25cm}
	
	\begin{thm}
		There exists a 9-tuple of optimal controls $(\mu_{1A}^{*}(t), \mu_{2A}^{*}(t), \mu_{3A}^{*}(t),
		\mu_{2Rem}^{*}(t), \mu_{3Rem}^{*}(t),\\
		\mu_{2INF}^{*}(t), \mu_{3INF}^{*}(t), \mu_{2Lop/Rit}^{*}(t), \mu_{3Lop/Rit}^{*}(t))$ in the set of admissible controls $U$ such that the cost functional is maximized i.e., 
		
		$J = \max_{(\mu_{1A}, \mu_{2A}, \mu_{3A}, \mu_{2Rem}, \mu_{3Rem}, \mu_{2INF}, \mu_{3INF}, \mu_{2Lop/Rit}, \mu_{3Lop/Rit}) \in U} \bigg \{ J[\mu_{1A}(t), \mu_{2A}(t), \mu_{3A}(t)),
		\mu_{2Rem}(t)), \mu_{3Rem}(t)),\\
		\mu_{2INF}(t)), \mu_{3INF}(t)), \mu_{2Lop/Rit}(t)), \mu_{3Lop/Rit}(t)]\bigg\}$
		corresponding to the optimal control problem (\ref{sec2equ1}) - (\ref{temp}).
	\end{thm}
	
	\begin{proof}
		
		In order to show the existence of optimal control functions, we will show that the following conditions are satisfied : 
		
		\begin{enumerate}
			\item  The solution set for the system (\ref{sec2equ1}) - (\ref{temp}) along with bounded controls must be non-empty, $i.e.$, $\Omega \neq \phi$.
			
			\item  U is closed and convex and system should be expressed linearly in terms of the control variables with coefficients that are functions of time and state variables.
			
			\item The Lagrangian L should be convex on U and $L(S, I, V, \mu_{1A}, \mu_{2A}, \mu_{3A}, \mu_{2Rem}, \mu_{3Rem}, \mu_{2zinc}, \mu_{3zinc}, \\
			\mu_{2Lop/Rit}, \mu_{3Lop/Rit}) \geq g(\mu_{1A}, \mu_{2A}, \mu_{3A},  \mu_{2Rem}, \mu_{3Rem}, \mu_{2INF}, \mu_{3INF}, \mu_{2Lop/Rit}, \mu_{3Lop/Rit})$, where $g$ is a continuous function of control variables such that $|(\mu_{1A}, \mu_{2A}, \mu_{3A}, \mu_{2Rem}, \mu_{3Rem}, \mu_{2INF}, \mu_{3INF}, \mu_{2Lop/Rit}, \mu_{3Lop/Rit})|^{-1} \\ g(\mu_{1A}, \mu_{2A}, \mu_{3A}, \mu_{2Rem}, \mu_{3Rem}, \mu_{2INF}, \mu_{3INF}, \mu_{2Lop/Rit}, \mu_{3Lop/Rit}) \to \infty$ \\ whenever  $|(\mu_{1A}, \mu_{2A}, \mu_{3A}, \mu_{2Rem}, \mu_{3Rem}, \mu_{2INF}, \mu_{3INF}, \mu_{2Lop/Rit}, \mu_{3Lop/Rit})| \to \infty$, where $|.|$ is an $l^2(0,T)$ norm.
		\end{enumerate}
		
		Now we will show that each of the conditions are satisfied : 
		
		1. From positivity and boundedness of solutions of the system  (\ref{sec2equ1}) - (\ref{sec2equ3}), all solutions are bounded for each bounded control variable in $U$.
		
		Also, the right hand side of the system (\ref{sec2equ1}) - (\ref{sec2equ3}) satisfies Lipschitz condition with respect to state variables. 
		
		Hence, using the positivity and boundedness condition and the existence of solution from Picard-Lindelof Theorem \cite{makarov2013picard}, we have satisfied condition 1.
		
		2. $U$ is closed and convex by definition. Also, the system (\ref{sec2equ1}) - (\ref{sec2equ3}) is clearly linear with respect to controls such that coefficients are only state variables or functions dependent on time. Hence condition 2 is satisfied.
		
		3. Choosing $g(\mu_{1A}, \mu_{2A}, \mu_{3A},   \mu_{2Rem}, \mu_{3Rem}, \mu_{2INF}, \mu_{3INF}, \mu_{2Lop/Rit}, \mu_{3Lop/Rit}) \\= c(\mu_{1A}, \mu_{2A}, \mu_{3A},  \mu_{2Rem}, \mu_{3Rem}, \mu_{2INF}, \mu_{3INF}, \mu_{2Lop/Rit}, \mu_{3Lop/Rit})$ such that $c = min\left\{A_{1}, A_{2},A_{3}, A_{4}\right\}$, we can satisfy the condition 3.
		
		Hence there exists a control 9-tuple $(\mu_{1A}, \mu_{2A}, \mu_{3A},  \mu_{2Rem}, \mu_{3Rem}, \mu_{2INF}, \mu_{3INF}, \mu_{2Lop/Rit}, \mu_{3Lop/Rit})\in U$ \\ that maximizes the cost function (\ref{temp}).
	\end{proof} \vspace{.25cm}
	
	\section{Characteriztion of Optimal Controls}\vspace{.25cm}
	
	We now obtain the necessary conditions for optimal control functions using the Pontryagin's Maximum Principle \cite{liberzon2011calculus} and also obtain the characteristics of the optimal controls.
	
	The Hamiltonian for this problem is given by \\
	
	$H(S, I , V, \mu_{1A}, \mu_{2A}, \mu_{3A}, \mu_{2Rem}, \mu_{3Rem}, \mu_{2INF}, \mu_{3INF}, \mu_{2Lop/Rit}, \mu_{3Lop/Rit},\lambda) \\
	= L(S, I, V, \mu_{1A}, \mu_{2A}, \mu_{3A}, \mu_{2Rem}, \mu_{3Rem}, \mu_{2INF}, \mu_{3INF}, \mu_{2Lop/Rit}, \mu_{3Lop/Rit}) + \lambda_{1} \frac{\mathrm{d} S}{\mathrm{d} t} +\lambda _{2}\frac{\mathrm{d} I}{\mathrm{d} t}+ \lambda _{3} \frac{\mathrm{d} V}{\mathrm{d} t}$\\
	
	Here $\lambda$ = ($\lambda_{1}$,$\lambda_{2}$,$\lambda_{3}$) is called co-state vector or adjoint vector.
	
	Now the Canonical equations that relate the state variables to the co-state variables are  given by 
	
	\begin{equation}
		\begin{aligned}
			\frac{\mathrm{d} \lambda _{1}}{\mathrm{d} t} &= -\frac{\partial H}{\partial S}\\
			\frac{\mathrm{d} \lambda _{2}}{\mathrm{d} t} &= -\frac{\partial H}{\partial I}\\
			\frac{\mathrm{d} \lambda _{3}}{\mathrm{d} t} &= -\frac{\partial H}{\partial V}
		\end{aligned}
	\end{equation}
	
	Substituting the Hamiltonian value gives the canonical system 
	
	\begin{equation}
		\begin{aligned}
			\frac{\mathrm{d} \lambda _{1}}{\mathrm{d} t} &= \lambda _{1}(\beta V+\mu+\mu_{1A})-\lambda _{2} \beta V\\
			\frac{\mathrm{d} \lambda _{2}}{\mathrm{d} t} &= 1+\lambda _{2}\bigg(x+(\mu_{2A}+\mu_{2Rem}+\mu_{2INF}+\mu_{2Lop/Rit}+\mu)\bigg)-\lambda _{3} \bigg(\alpha -(\mu_{3A}+\mu_{3Rem}+\mu_{3INF}+\mu_{3Lop/Rit})\bigg)\\
			\frac{\mathrm{d} \lambda _{3}}{\mathrm{d} t} &= 1+\lambda _{1}\beta S-\lambda _{2}\beta S+\lambda _{3}(y+\mu_{1})
		\end{aligned}
	\end{equation}
	
	where $x=d_{1}+d_{2}+d_{4}+d_{5}+d_{6}$and $y=b_{1}+b_{2}+b_{4}+b_{5}+b_{6}$ 	along with transversality conditions
	$ \lambda _{1} (T) = 0, \  \lambda _{2} (T) = 0, \  \lambda _{3} (T) = 0. $
	
		Now, to obtain the optimal controls, we will use the Hamiltonian minimization condition 
	$ \frac{\partial H}{\partial u_{i}}$ = 0 , at  $\mu^{*}.$
	
	Differentiating the Hamiltonian and solving the equations, we obtain the optimal controls as 
	
	\begin{eqnarray*}
		\mu_{1A}^{*} &=& \min\bigg\{ \max\bigg\{\frac{\lambda _{1}S}{2A_{1}},0 \bigg\}, u_{1A}max\bigg\}\\
		\mu_{2A}^{*} &= &\min\bigg\{ \max\bigg\{\frac{\lambda _{2}I}{2A_{1}},0 \bigg\}, u_{2A}max\bigg\}\\
		\mu_{3A}^{*}& = &\min\bigg\{ \max\bigg\{\frac{\lambda _{3}I}{2A_{1}},0 \bigg\}, u_{3A}max\bigg\}\\
		\mu_{2Rem}^{*} &= &\min\bigg\{ \max\bigg\{\frac{\lambda _{2}I}{2A_{2}},0 \bigg\}, u_{2Rem}max\bigg\}\\
		\mu_{3Rem}^{*}& = &\min\bigg\{ \max\bigg\{\frac{\lambda _{3}I}{2A_{2}},0 \bigg\}, u_{3Rem}max\bigg\}\\
		\mu_{2INF}^{*} &= &\min\bigg\{ \max\bigg\{\frac{\lambda _{2}I}{2A_{3}},0 \bigg\}, u_{2INF}max\bigg\}\\
		\mu_{3INF}^{*}& = &\min\bigg\{ \max\bigg\{\frac{\lambda _{3}I}{2A_{3}},0 \bigg\}, u_{3INF}max\bigg\}\\
		\mu_{2Lop/Rit}^{*} &= &\min\bigg\{ \max\bigg\{\frac{\lambda _{2}I}{2A_{4}},0 \bigg\}, u_{2Lop/Rit}max\bigg\}\\
		\mu_{3Lop/Rit}^{*}& = &\min\bigg\{ \max\bigg\{\frac{\lambda _{3}I}{2A_{4}},0 \bigg\}, u_{3Lop/Rit}max\bigg\}
	\end{eqnarray*}

\section{Optimal Drug Regimen} \vspace{.25cm}

In this section, we perform numerical simulations to understand the efficacy of single and multiple drug interventions and propose the optimal drug regimen in these scenarios. This is done by studying the effect of the corresponding controls  on the  dynamics of the system (\ref{sec2equ1}) - (\ref{sec2equ3}). 

The efficacy of various combinations of controls considered are:

\begin{itemize}
	\item [1.] Single drug/intervention administration.
	\item [2.] Two drugs/interventions administration.
	\item [3.] Three drugs/interventions administration.
	\item [4.] All the four drugs/interventions administration.
\end{itemize}	

The theoretical results obtained  are validated for a  set of model parameters obtained using $ode23$ solver in MATLAB. For our simulations, we have taken the total number of days as $T =30$. The values of $\omega,$ $\mu$ and $\mu_{1},$ $\alpha$ are taken  from \cite{2} and \cite{ea2020host} respectively.  The rest of the  parameter values of the model are estimated minimizing the root mean square difference between the model predictive output and the experimental data chosen from \cite{wolfel2020virological, qin2020dysregulation}.  Based on the above, the parameter values chosen for the model 1  are summarized in the following table.\\

\begin{table}[ht!]
	\begin{center}
		\begin{tabular}{|c|c|c|c|c|c|c|c|c|c|c|c|c|c|c|c|c|}
			\hline
			$\omega$ & $\beta$ & $\mu$ & $\mu_{1}$ & $\alpha$ & $d_{1}$ & $d_{2}$ & $d_{4}$ & $d_{5}$ & $d_{6}$ &  $b_{1}$ & $b_{2}$  & $b_{4}$ & $b_{5}$ & $b_{6}$ \\
			\hline
			10 & 0.05& .1 & 1.1 & .5 & 0.027& 0.22  &0.428 & 0.01 & 0.01 & 0.1& .1 & .11 & .1 & .07  \\
			\hline
		\end{tabular}
	\end{center}
\end{table}

We first solve the state system numerically using Fourth Order Runge-Kutta method in MATLAB without any interventions. We take the initial values of state variables to be $S(0) = 3.2 \times 10^5, I(0) = 0, V(0) = 5$ and  the initial values of the control parameters as zeros.

Now, to simulate the system with controls, we use the Forward-Backward Sweep method starting with the initial values of controls to be zero and solve the state system forward in time. Following this we solve the adjoint state system backward in time due to the transversality conditions, using the optimal state variables and initial values of optimal controls which are zero4.

Now, using the values of adjoint state variables, the values of optimal control are updated and with these updated control variables, we go through this process again. We continue this till the convergence criterion is met \cite{liberzon2011calculus}. 

In Survival analysis, Hazard Ratio (HR) plays a crucial role in determining the rate at which the people treated by drug may suffer a certain complication per unit time as the control population. Larger the hazard ratio, more harmful the drug to be administered. We use this concept in assigning weights to our objective function in our  model. In the following table we enlist the  the hazard ratios for the four drugs considered in this work.

\begin{table}[htp!]
	\begin{center}
		\begin{tabular}{ | c | c | c | c | }
			\hline \hline
			\textbf{No.} & \textbf{Drug} & \textbf{HR} & \textbf{Source} \\
			\hline
			1 & Arbidol & 0.183 & \cite{arbidolHR} \\
			\hline
			2 & Remdesivir & 0.33 & \cite{remdesivirHR} \\
			\hline
			3 & Interferon & 0.375 & \cite{IFNhazard} \\
			\hline
			4 & Lopinavir/Ritonavir & 0.4 & \cite{lopinavirHR} \\
			\hline
		\end{tabular}
		\caption{Hazard Ratios (HR) for drugs considered}
	\end{center} \label{t0}
\end{table}

Based on the hazard ratio and effectiveness of the drugs in treatment of Covid-19 we choose the positive weights  for objective coefficients as $A_{1}$ = 500, $A_{2}$ = 400, $A_{3}$ = 250, $A_{4}$ = 200. $A_{1}$ is chosen high compared to other coefficients because it has the least hazard ratio.

\subsection{WITHOUT ANY DRUGS / INTERVENTIONS} \vspace{.25cm}

In this section we simulate the behavior of susceptible, infected  cells and viral load over time. As can be seen from figure \ref{c1}   the susceptible cells reduce and the infected cells increase exponentially due to the increase in viral load over a period of time.

\begin{figure}[hbt!]
	\begin{center}
		\includegraphics[width=4in, height=2.5in, angle=0]{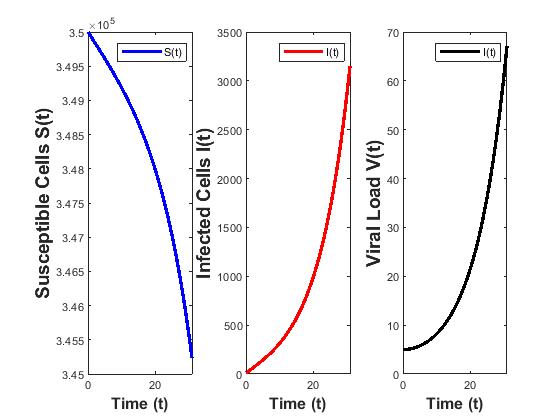}
		\caption{ Figure depicting the  $S, I, V$ populations   without control interventions over the time. The exponential growth of the infected cells and viral load can be observed. }  \label{c1}
	\end{center}
\end{figure}

\vspace{.25cm}

\subsection{SINGLE DRUG / INTERVENTION} \vspace{.25cm}

In this section we study the dynamics  of susceptible, infected  cells and viral load when  each of these four drugs is administered individually.  Figures  \ref{c2}, \ref{c3}, \ref{c4} depict the susceptible, infected population and viral load.  \vspace{.25cm}

\begin{table}[ht!]
	\centering 
	\begin{tabular}{|l|l|l|l|} 
		\hline
		\textbf{Drug Combinations} &  \textbf{Avg Susceptible cells}&  \textbf{Avg Infected cells}& \textbf{Avg Viral load} \\ 
		\hline 
		$U_1=0,U_2=0,U_3=0,U_4=U_4^*$ & 3.4909 $\times$ 10$^5$ & 325.9436 & 6.8030 \\
		\hline
		$U_1=0,U_2=0,U_3=U_3^*,U_4=0$ & 3.4908 $\times$ 10$^5$ & 334.6489   & 7.0864   \\
		\hline
		$U_1=0,U_2=U_2^*,U_3=0,U_4=0$ & 3.4907 $\times$ 10$^5$ & 348.6319 & 7.5458 \\
		\hline
		$U_1=U_1^*,U_2=0,U_3=0,U_4=0$ & 3.4901 $\times$ 10$^5$ & 353.5103 & 7.7088\\
		\hline
		$U_1=0,U_2=0,U_3=0,U_4=0$ &3.4904 $\times$ 10$^5$& 373.8252& 8.4191 \\ 
		\hline
	\end{tabular} 
	\caption{Table depicting the average values of the susceptible cells, infected cells and the viral load with respect to each of these drug interventions when administered individually. } \label{t1}
\end{table} \vspace{.25cm}	

\begin{figure}[hbt!]
	\begin{center}
		\includegraphics[width=4in, height=2.4in, angle=0]{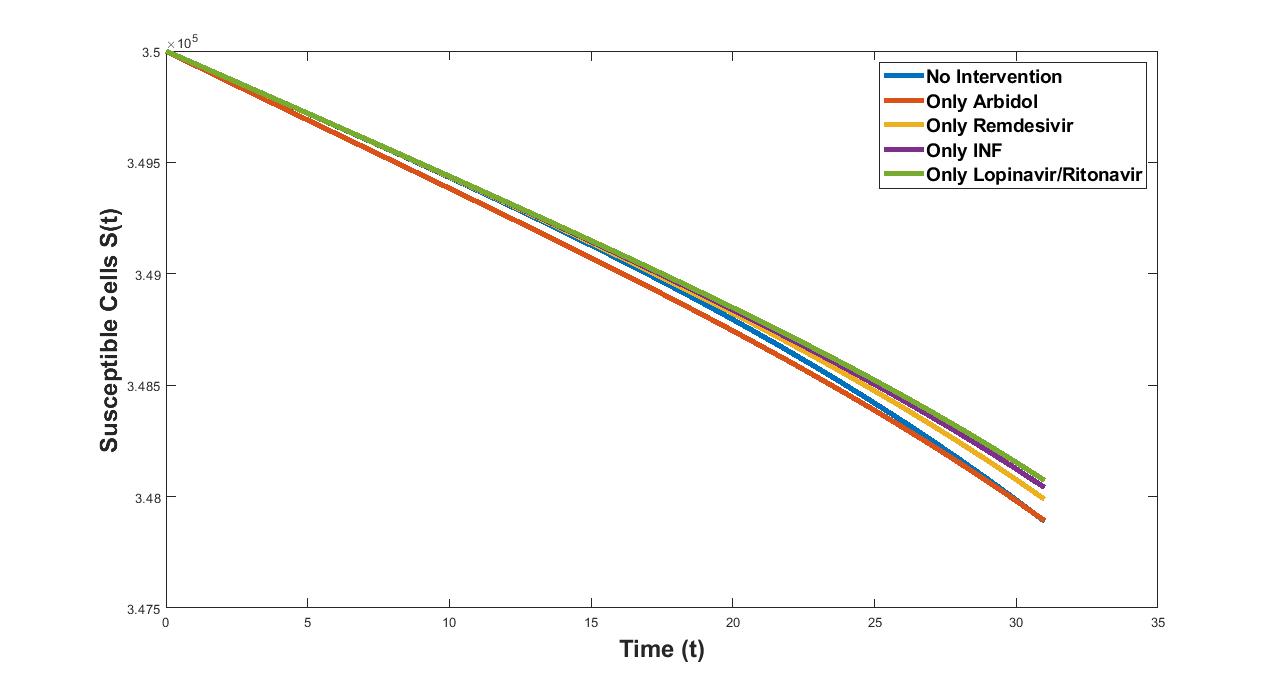}
		\caption{    Figure depicting the dynamics of Susceptible cells ($S$)   under the optimal controls $U_{1}^{*}, U_{2}^{*}$, $U_{3}^{*},U_{4}^{*}.$ Each curve represents the dynamics  of susceptible cells with respect to  a single control intervention.} \label{c2}
	\end{center}
\end{figure}

\begin{figure}[hbt!]
	\begin{center}
		\includegraphics[width=4in, height=2.4in, angle=0]{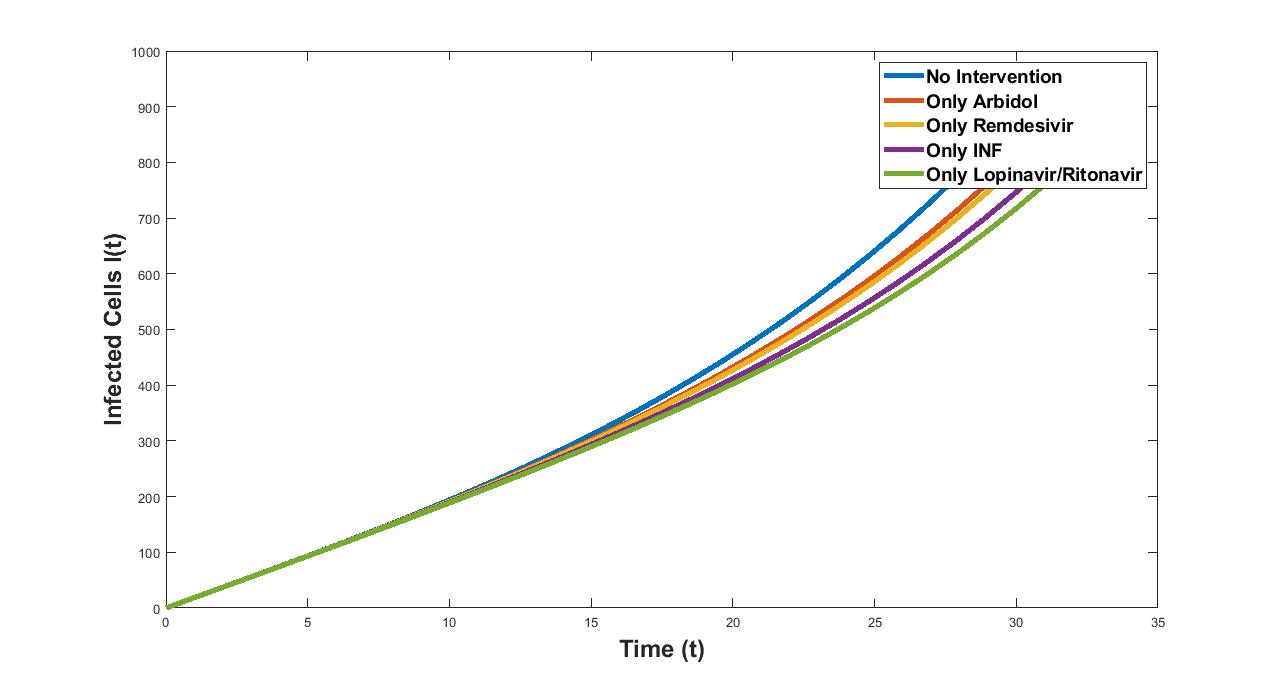}
		\caption{   Figure depicting the dynamics of Infected cells ($I$)   under the optimal controls $U_{1}^{*}, U_{2}^{*}$, $U_{3}^{*},U_{4}^{*}.$ Each curve represents the dynamics  of infected cells with respect to a single control intervention.} \label{c3}
	\end{center}
\end{figure}

\begin{figure}[hbt!]
	\begin{center}
		\includegraphics[width=4in, height=2.4in, angle=0]{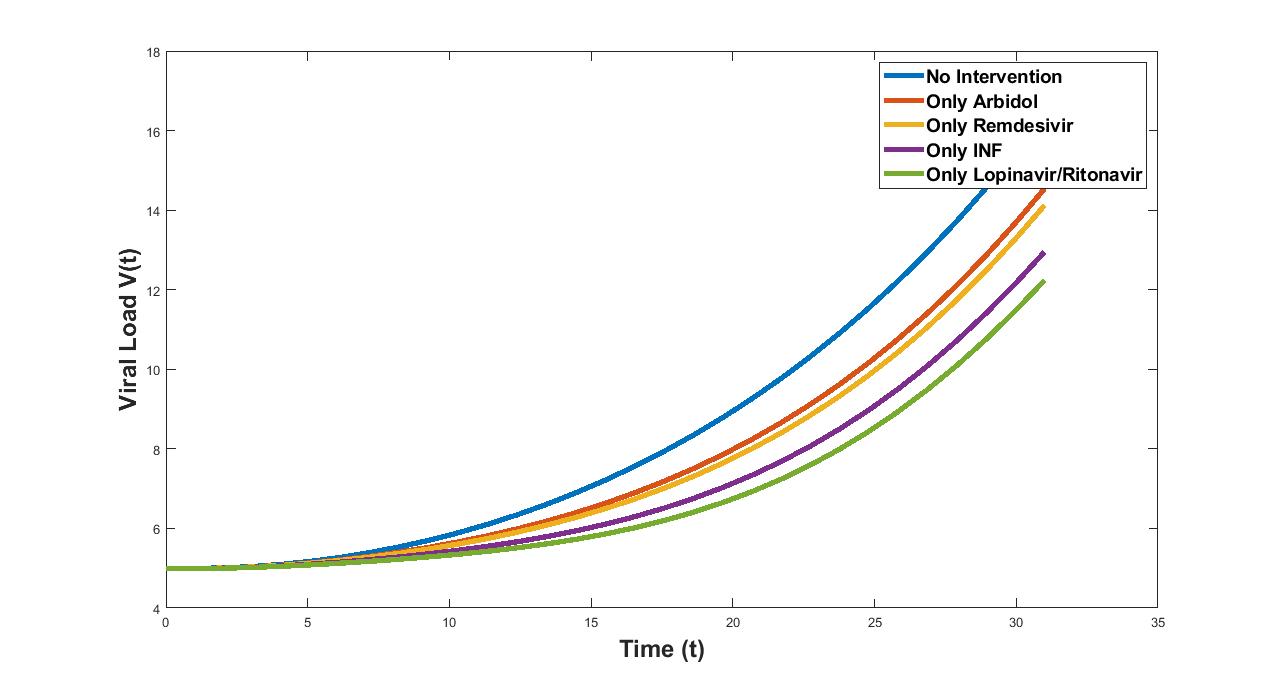}
		\caption{   Figure depicting the dynamics of Viral load ($I$)   under the optimal controls $U_{1}^{*}, U_{2}^{*}$, $U_{3}^{*},U_{4}^{*}.$ Each curve represents the dynamics  of viral load with respect to  a  single  control intervention. } \label{c4}
	\end{center}
\end{figure} \vspace{.75cm}

From the  table {\ref{t1}}  the average values of the susceptible cells, infected cells and the viral load with respect to each of these drug interventions when administered individually is listed. From this table {\ref{t1}} it can be seen that the drug Lopinavir/Ritonavir ($U_4 = U_4^*$) reduces the infected cells and viral load the best compared to other drugs when administered  individually followed by drug INF ($U_3 = U_3^*$). On the other hand drug Arbidol ($U_1 = U_1^*$) does the best job in reducing the susceptible cells. \vspace{.25cm}

\subsection{TWO DRUGS / INTERVENTIONS}  \vspace{.25cm}

In this section we study the dynamics  of susceptible, infected  cells and viral load when   two  control interventions are administered at a time.  Figures  \ref{c5}, \ref{c6}, \ref{c7} depict the susceptible, infected population and viral load.  \vspace{.4cm}

\begin{table}[ht!]
	\centering 
	\begin{tabular}{|l|l|l|l|} 
		\hline
		\textbf{Drug Combinations} &  \textbf{Avg Susceptible cells}&  \textbf{Avg Infected cells}& \textbf{Avg Viral load} \\
		\hline
		$U_1=0,U_2=0,U_3=U_3^*,U_4=U_4^*$ & 3.4912 $\times$ 10$^5$ & 295.1982 & 5.8179 \\
		\hline
		$U_1=0,U_2=U_2^*,U_3=0,U_4=U_4^*$ & 3.4911 $\times$ 10$^5$ & 306.0408 & 6.1623  \\
		\hline
		$U_1=U_1^*,U_2=0,U_3=0,U_4=U_4^*$ &3.4906 $\times$ 10$^5$& 309.8094 & 6.2838 \\ 
		\hline
		$U_1=0,U_2=U_2^*,U_3=U_3^*,U_4=0$ & 3.4911 $\times$ 10$^5$ & 313.7398 & 6.4089\\
		\hline
		$U_1=U_1^*,U_2=0,U_3=U_3^*,U_4=0$ &3.4905 $\times$ 10$^5$ & 317.6816 & 6.5368 \\
		\hline
		$U_1=U_1^*,U_2=U_2^*,U_3=0,U_4=0$ &3.4904 $\times$ 10$^5$ & 330.2106& 6.9428 \\ 
		\hline
		$U_1=U_2=U_3=U_4=0$ & 3.4904 $\times$ 10$^5$ & 374.8252& 8.4191 \\
		\hline
	\end{tabular}
	\caption{Table depicting the average values of the susceptible cells, infected cells and the viral load with respect to  two  control interventions administered at a time. } \label{t2}
\end{table} \vspace{.4cm}	

From  the  table {\ref{t2}}  the average values of the susceptible cells, infected cells and the viral load with respect to   two  control interventions administered at a time  is listed. From this table {\ref{t2}} it can be seen that the combination of INF and  Lopinavir/Ritonavir  ($U_3 = U_3^*, U_4 = U_4^*$) reduces the infected cells and viral load the best compared to other combination of drugs when administered   followed by drug combination Remdesivir and  Lopinavir/Ritonavir ($U_2 = U_2^*, U_4 = U_4^*$). On the other hand,  the drug combination of Arbidol and
Remdeivir $(U_1 = {{U_1}^*}, U_2 = {U_2}^*) $ does the
best job in reducing the susceptible cells. \vspace{.25cm}

\begin{figure}[hbt!]
	\begin{center}
		\includegraphics[width=4in, height=2.4in, angle=0]{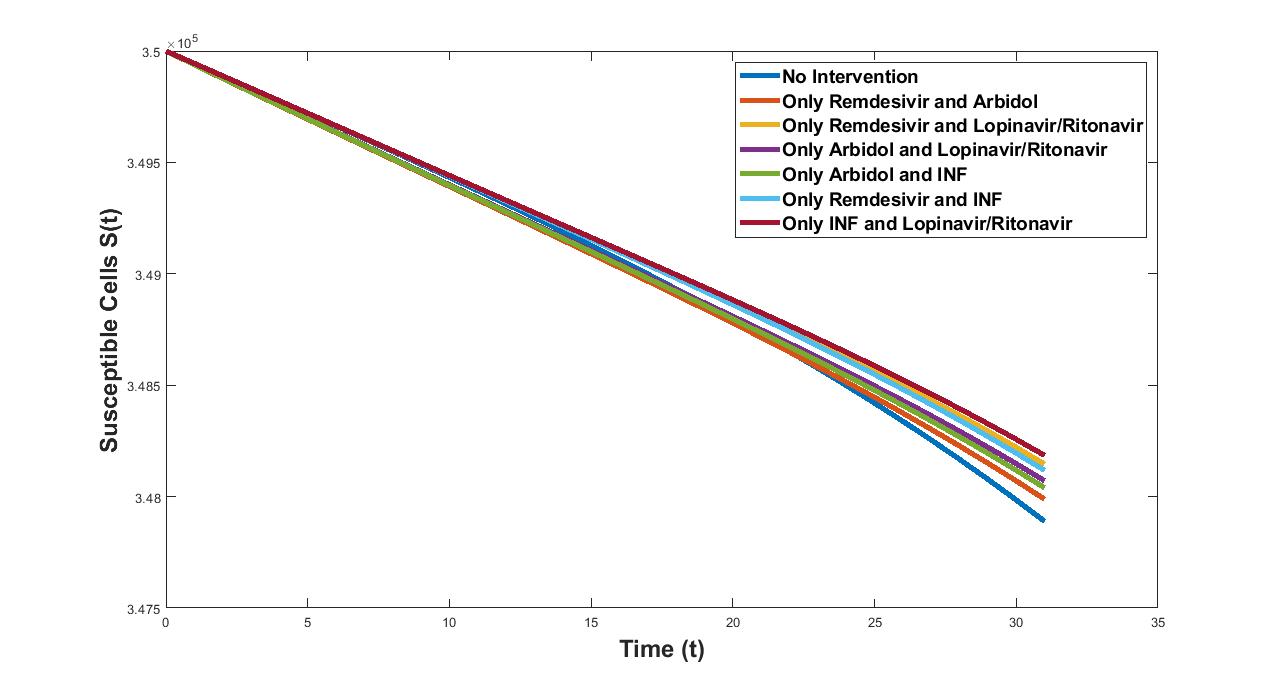}
		\caption{    Figure depicting the dynamics of Susceptible cells ($S$)   under the optimal controls $U_{1}^{*}, U_{2}^{*}$, $U_{3}^{*},U_{4}^{*}.$ Each curve represents the dynamics  of susceptible cells with respect to  two  control interventions administered at a time.} \label{c5}
	\end{center}
\end{figure}

\begin{figure}[hbt!]
	\begin{center}
		\includegraphics[width=4in, height=2.4in, angle=0]{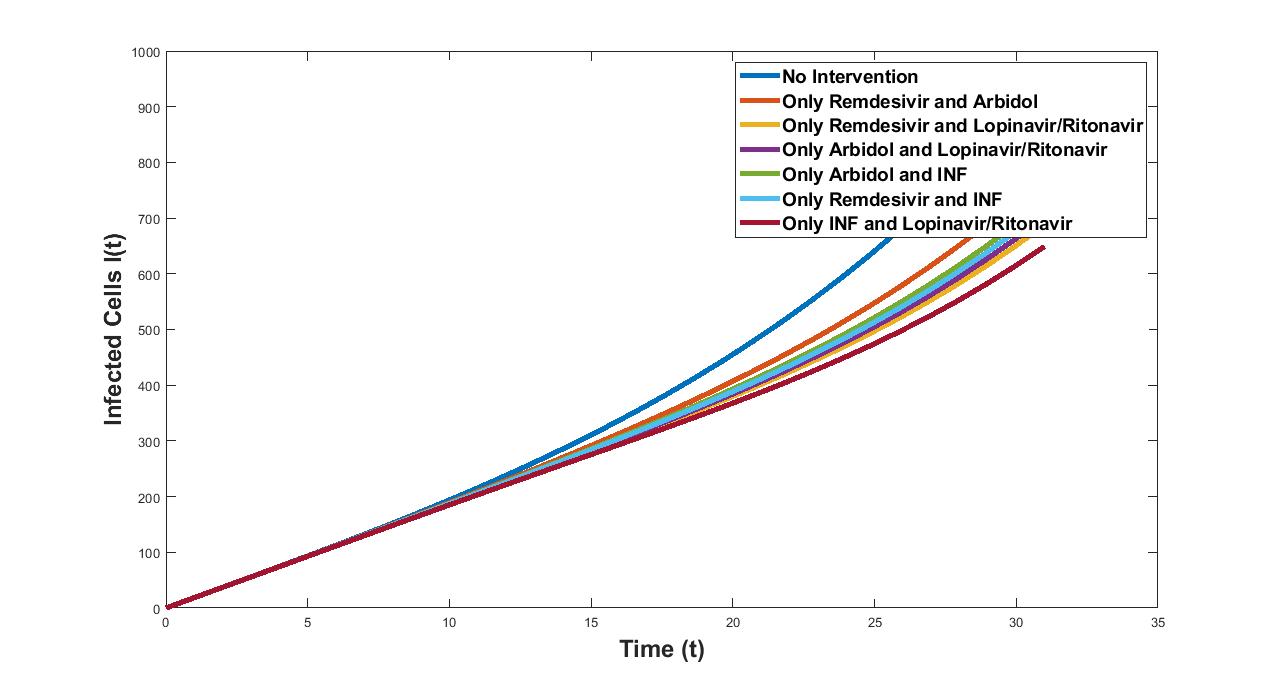}
		\caption{   Figure depicting the dynamics of Infected cells ($I$)   under the optimal controls $U_{1}^{*}, U_{2}^{*}$, $U_{3}^{*},U_{4}^{*}.$ Each curve represents the dynamics  of infected cells with respect to  two  control interventions administered at a time.} \label{c6}
	\end{center}
\end{figure}

\begin{figure}[hbt!]
	\begin{center}
		\includegraphics[width=4in, height=2.4in, angle=0]{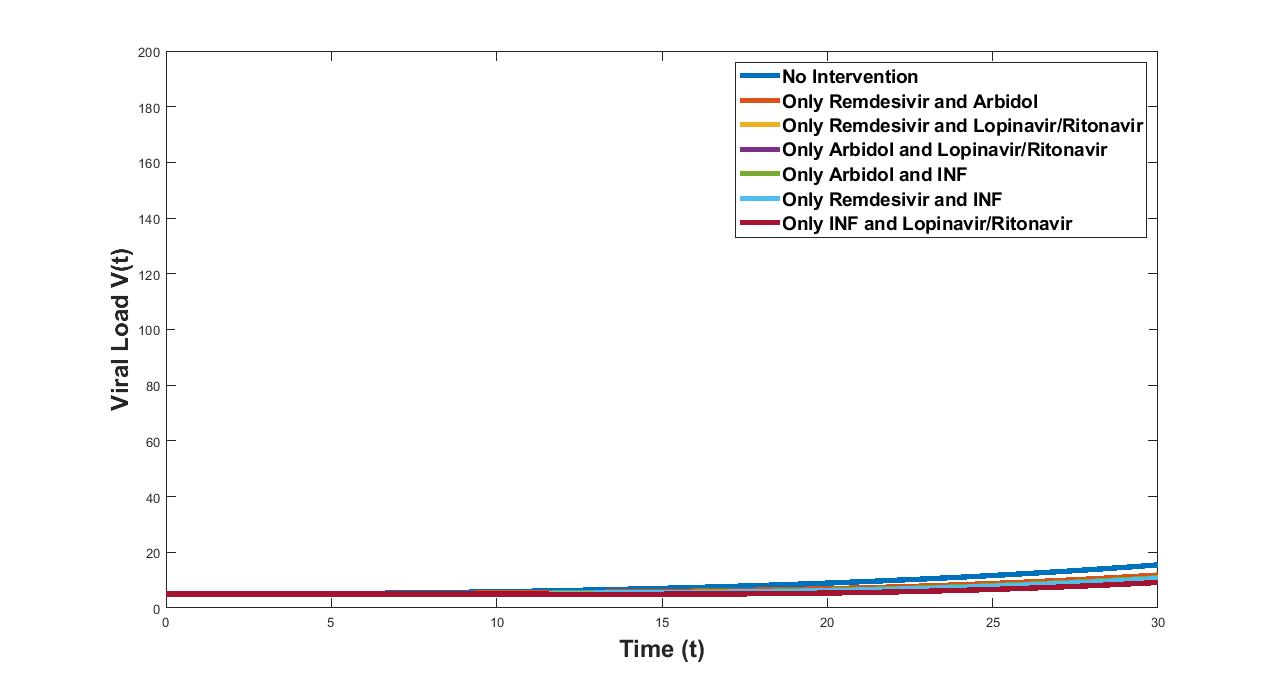}
		\caption{   Figure depicting the dynamics of Viral load ($I$)   under the optimal controls $U_{1}^{*}, U_{2}^{*}$, $U_{3}^{*},U_{4}^{*}.$ Each curve represents the dynamics  of viral load with respect to  two  control interventions administered at a time. } \label{c7}
	\end{center}
\end{figure}  \vspace{.75cm}

\subsection{THREE DRUGS / INTERVENTIONS}  \vspace{.25cm}

In this section we study the dynamics  of susceptible, infected  cells and viral load when   three  control interventions are administered at a time.  Figures  \ref{c8}, \ref{c9}, \ref{c10} depict the susceptible, infected population and viral load.  \vspace{.4cm}

\begin{table}[ht!]
	\centering 
	\begin{tabular}{|l|l|l|l|} 
		\hline
		\textbf{Drug Combinations} &  \textbf{Avg Susceptible cells}&  \textbf{Avg Infected cells}& \textbf{Avg Viral load}\\
		\hline
		$U_1=U_{1}^{*},U_2=0,U_3=U_{3}^{*},U_4=U_{4}^{*}$ & 3.4910 $\times$ 10$^5$ & 281.8542 & 5.3993  \\
		\hline
		$U_1=U_{1}^{*},U_2=U_{2}^{*},U_3=0,U_4=U_{4}^{*}$ & 3.4909 $\times$ 10$^5$ & 291.7824& 5.7105 \\
		\hline
		$U_1=U_{1}^{*},U_2=U_{2}^{*},U_3=U_{3}^{*},U_4=0$ & 3.4908 $\times$ 10$^5$ & 298.7510& 5.9310  \\
		\hline
		$U_1=0,U_2=U_{2}^{*},U_3=U_{3}^{*},U_4=U_{4}^{*}$ & 3.4911 $\times$ 10$^5$& 307.0594& 6.1791 \\
		\hline
		$U_1=0,U_2=0,U_3=0,U_4=0$ & 3.4904 $\times$ 10$^5$ & 374.8252 & 8.4191 \\
		\hline
	\end{tabular} 
	\caption{Table depicting the average values of the susceptible cells, infected cells and the viral load with respect to  two  three control interventions administered at a time. } \label{t3}
\end{table} \vspace{.4cm}	

From  the  table {\ref{t3}}  the average values of the susceptible cells, infected cells and the viral load with respect to   three  control interventions administered at a time  is listed. From this table {\ref{t3}} it can be seen that the combination of Arbidol, INF and Lopinavir/Ritonavir ($U_1 = U_1^*, U_3 = U_3^*, U_4 = U_4^*$) reduces the infected cells and viral load the best compared to other combination of drugs when administered   followed by drug combination Arbidol, Remdsivir and  Lopinavir/Ritonavir ($U_1 = U_1^*, U_2 = U_2^*, U_4 = U_4^*$). On the other hand, surprisingly  the average number of the susceptible cells are lesser in no control intervention case compared to any other case having three control interventions. This can be attributed to side effects of multiple drug interventions involving more than two drug administration at a time to patient. \vspace{.4cm}

\begin{figure}[hbt!]
	\begin{center}
		\includegraphics[width=4in, height=2.4in, angle=0]{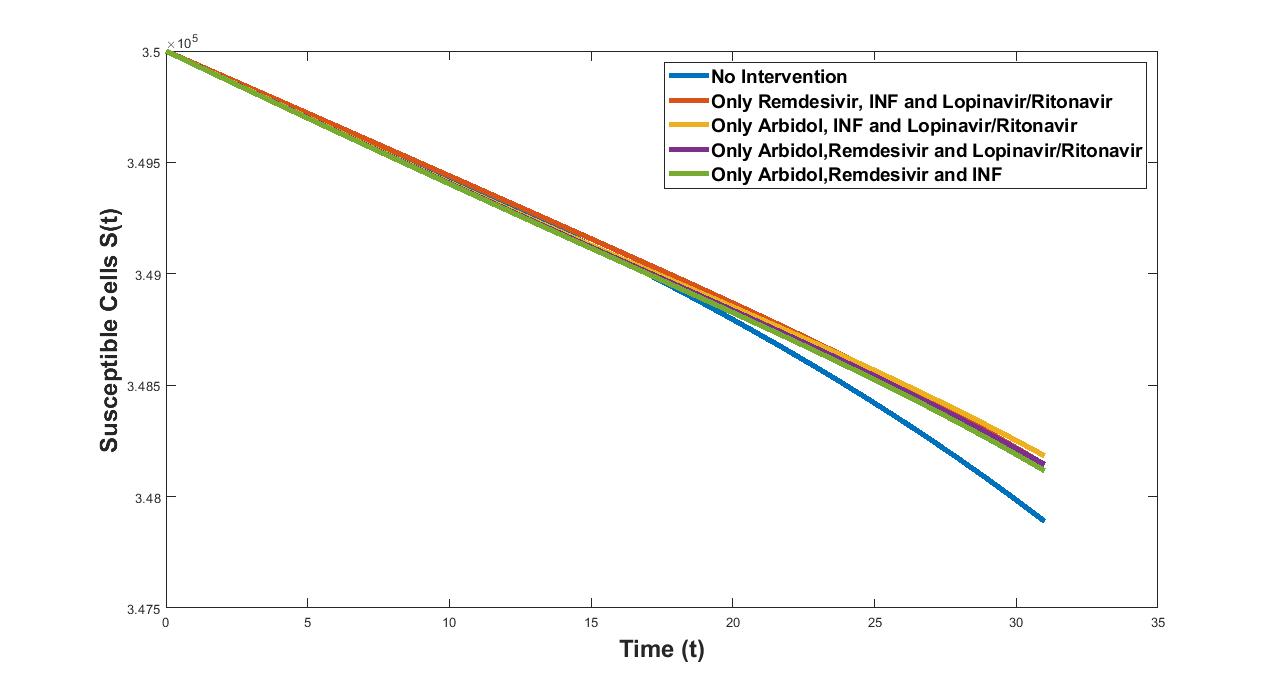}
		\caption{    Figure depicting the dynamics of Susceptible cells ($S$)   under the optimal controls $U_{1}^{*}, U_{2}^{*}$, $U_{3}^{*},U_{4}^{*}.$ Each curve represents the dynamics  of susceptible cells with respect to  three  control interventions administered at a time.} \label{c8}
	\end{center}
\end{figure}

\begin{figure}[hbt!]
	\begin{center}
		\includegraphics[width=4in, height=2.4in, angle=0]{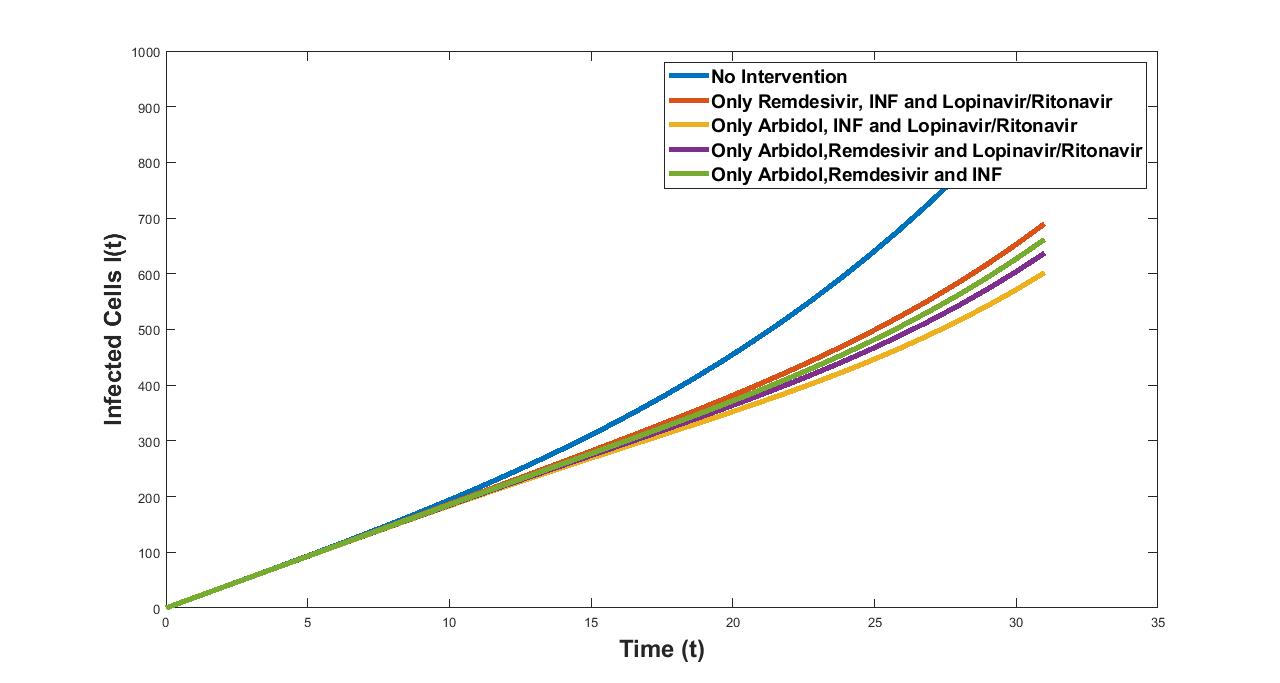}
		\caption{   Figure depicting the dynamics of Infected cells ($I$)   under the optimal controls $U_{1}^{*}, U_{2}^{*}$, $U_{3}^{*},U_{4}^{*}.$ Each curve represents the dynamics  of infected cells with respect to  three  control interventions administered at a time.}  \label{c9}
	\end{center}
\end{figure}

\begin{figure}[hbt!]
	\begin{center}
		\includegraphics[width=4in, height=2.4in, angle=0]{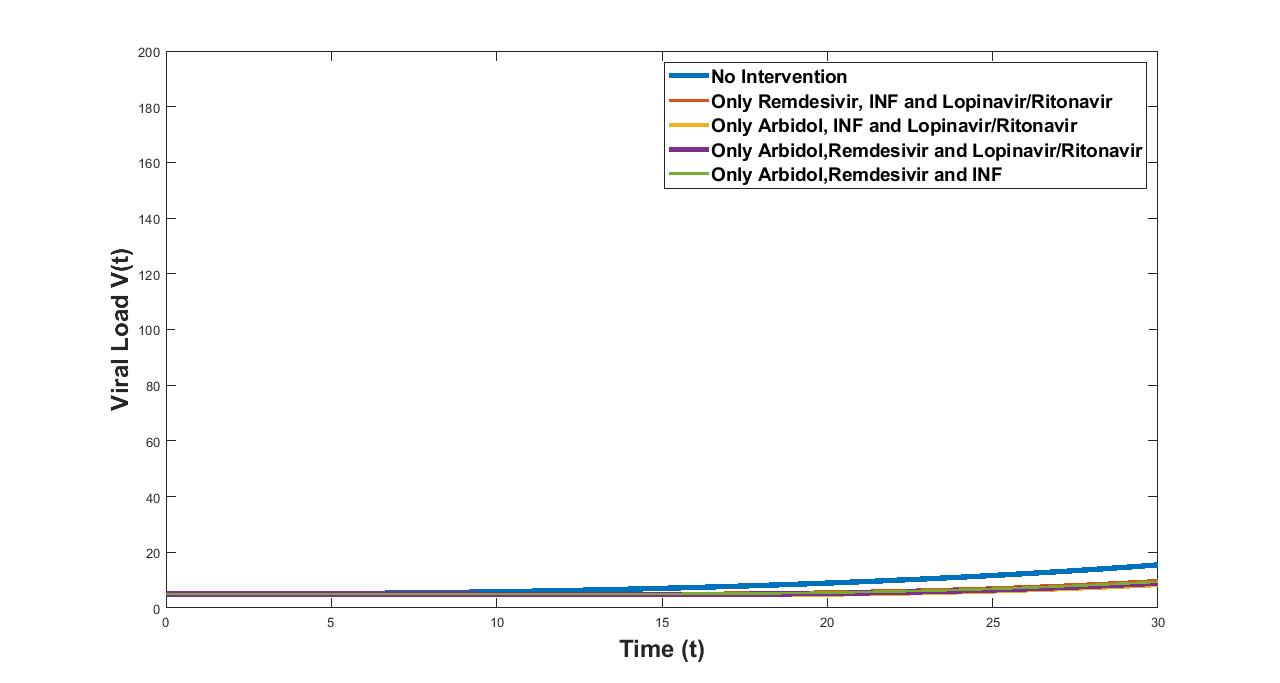}
		\caption{   Figure depicting the dynamics of Viral load ($I$)   under the optimal controls $U_{1}^{*}, U_{2}^{*}$, $U_{3}^{*},U_{4}^{*}.$ Each curve represents the dynamics  of viral load with respect to  three  control interventions administered at a time. }  \label{c10}
	\end{center}
\end{figure}

\vspace{.75cm}

\subsection{FOUR DRUGS / INTERVENTIONS}  \vspace{.25cm}

In this section we study the dynamics  of susceptible, infected  cells and viral load when   all the four  control interventions are administered at a time.  Figures  \ref{c11}, \ref{c12}, \ref{c13} depict the susceptible, infected population and viral load.  \vspace{.4cm}

\begin{table}[ht!]
	\centering 
	\begin{tabular}{|l|l|l|l|} 
		\hline
		\textbf{Drug Combinations} &  \textbf{Avg Suceptible cells}&  \textbf{Avg Infected cells}& \textbf{Avg Viral load}\\  
		\hline 
		$U_1=U_1^*,U_2=U_2^*,U_3=U_3^*,U_4=U_4^*$ & 3.4911 $\times $ 10$^5$ & 272.1884& 5.0925 \\ 
		\hline
		$U_1=U_2=U_3=U_4=0$ & 3.4904 $\times$ 10$^5$ & 374.8252 & 8.4191  \\
		\hline
	\end{tabular}
	\caption{Table depicting the average values of the susceptible cells, infected cells and the viral load with respect to  two  three control interventions administered at a time. } \label{t4}
\end{table}

\vspace{.4cm}

From  the  table {\ref{t4}}  the average values of the susceptible cells, infected cells and the viral load with respect to  all four control interventions administered at a time  is listed. From this table {\ref{t4}} it can be seen that the combination of all four drug interventions reduce the infected cells and viral load  compared to no intervention case. Like in the three drug control intervention case here also  the average number of the susceptible cells are lesser in no control intervention case compared to all four control interventions case. This again can be attributed to side effects of multiple drug interventions involving more than two drug administration at a time to patient. \vspace{.4cm}

\begin{figure}[hbt!]
	\begin{center}
		\includegraphics[width=4in, height=2.3in, angle=0]{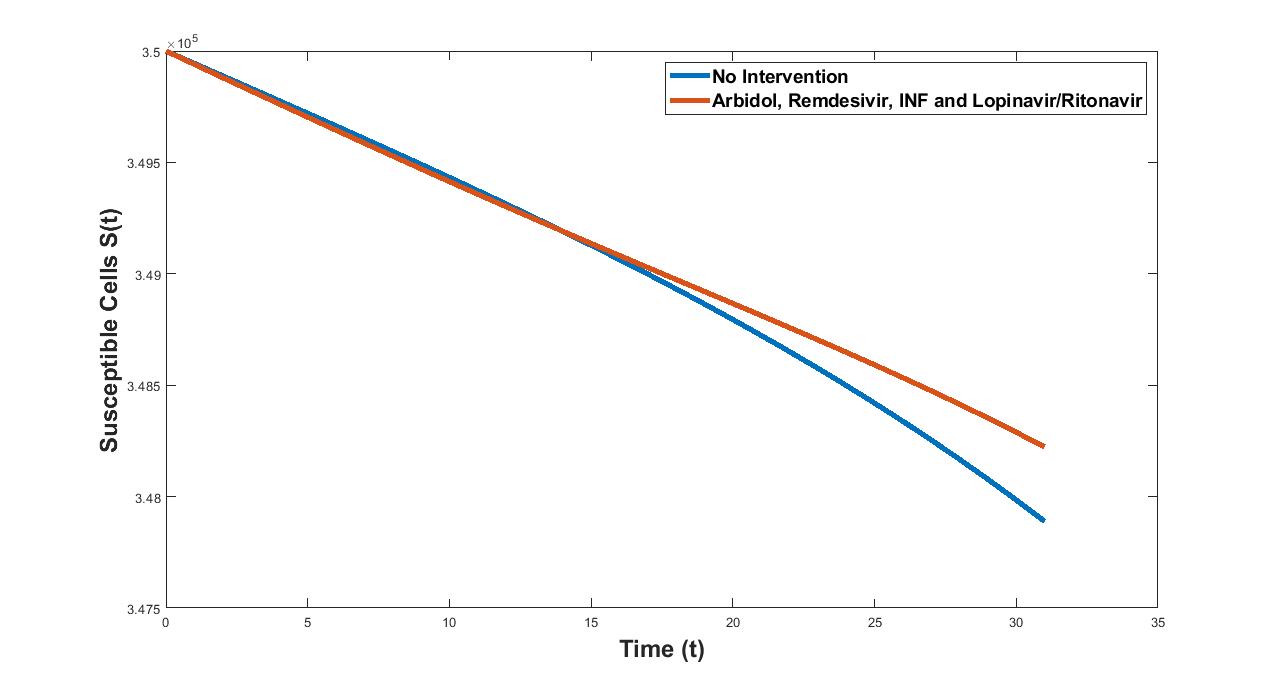}
		\caption{    Figure depicting the dynamics of Susceptible cells ($S$)   under the optimal controls $U_{1}^{*}, U_{2}^{*}$, $U_{3}^{*},U_{4}^{*}.$ Red curve represents the dynamics  of susceptible cells with respect to  all control interventions and blue curve represents susceptible cells for no intervention case.}  \label{c11}
	\end{center}
\end{figure}

\begin{figure}[hbt!]
	\begin{center}
		\includegraphics[width=4in, height=2.3in, angle=0]{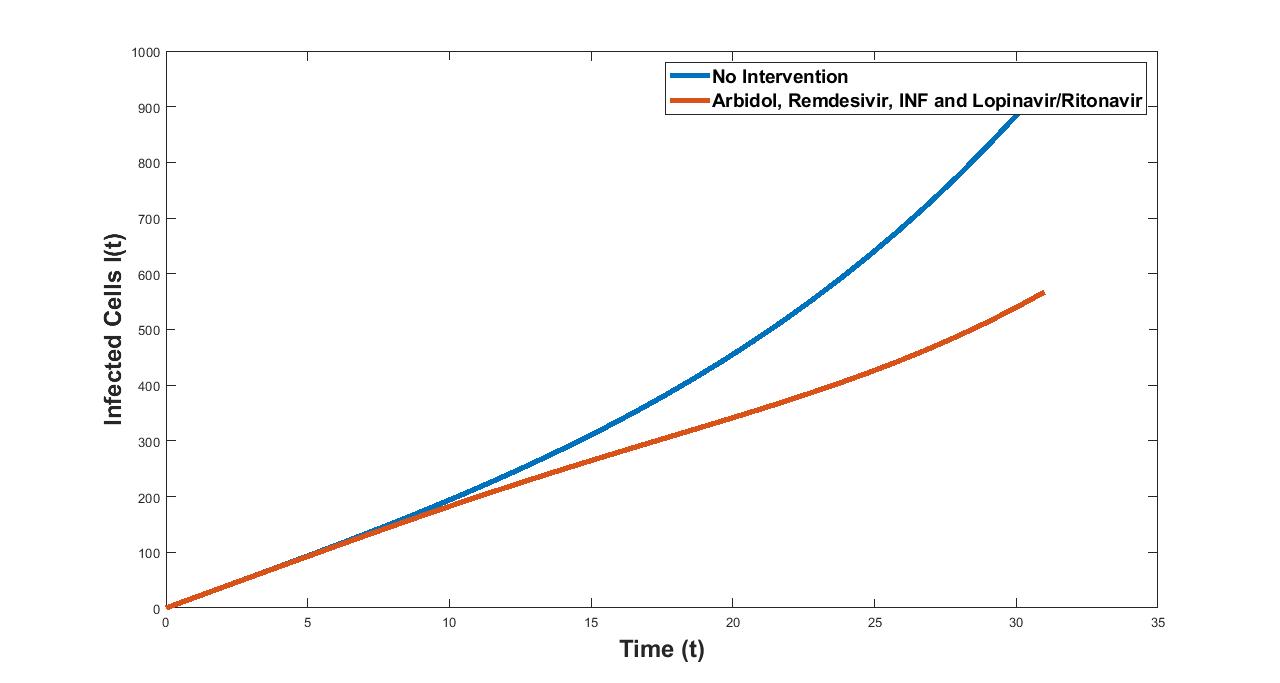}
		\caption{  Figure depicting the dynamics of Infected cells ($I$)   under the optimal controls $U_{1}^{*}, U_{2}^{*}$, $U_{3}^{*},U_{4}^{*}.$ Red curve represents the dynamics  of infected cells with respect to  all control interventions and blue curve represents infected cells for no intervention case.}  \label{c12}
	\end{center}
\end{figure}

\begin{figure}[hbt!]
	\begin{center}
		\includegraphics[width=4in, height=2.3in, angle=0]{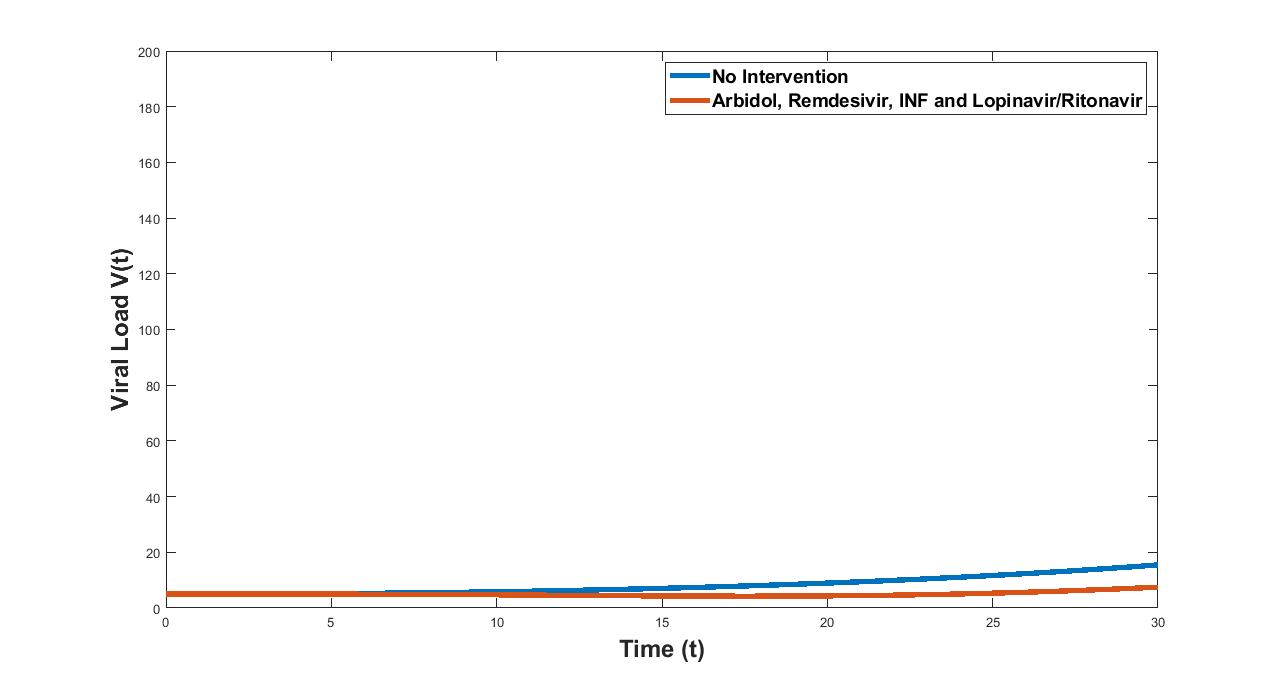}
		\caption{   Figure depicting the dynamics of Viral load ($I$)   under the optimal controls $U_{1}^{*}, U_{2}^{*}$, $U_{3}^{*},U_{4}^{*}.$ Red curve represents the dynamics  of viral load with respect to  all control interventions and blue curve represents viral load for no intervention case. } \label{c13}
	\end{center}
\end{figure}

\vspace{.75cm}

\section{Comparative effectiveness study} \vspace{.25cm}

In this section we do the comparative effectiveness study for the system 
\begin{eqnarray}
	\frac{dS}{dt}& =&  \omega \ - \beta SV   -\mu S  \label{sec6equ1} \\ 
	\frac{dI}{dt} &=& \beta SV \ -  { \bigg(d_{1}  + d_{2}  +  d_{4} + d_{5}+ d_{6}\bigg)I} -\mu I   \label{sec6equ2}\\ 
	\frac{dV}{dt} &=&  \alpha I   
	-\bigg( b_{1}  + b_{2}   +  b_{4}+ b_{5} + b_{6}\bigg)V    \ -  \mu_{1} V \label{sec6equ3}
\end{eqnarray}  \\
The basic reproductive number for the system (\ref{sec6equ1}) - (\ref{sec6equ3}) as obtained in \cite{ref1} is given by
\begin{equation}
	\mathcal{R}_{0} = \frac{ \beta \alpha \omega}{\mu (b_{1}+b_{2}+b_{4}+b_{5}+b_{6}+\mu_1)(d_{1}+d_{2}+d_{4}+d_{5}+d_{6}+\mu)}
	\num \label{eqn2}
\end{equation}

The disease-free equilibrium for the system can be seen to be
\begin{equation}
	E^{0} = \left( S^{0} , I^{0} , V^{0} \right) = \left( \frac{\omega}{\mu}, 0, 0 \right) \num \label{eqn3}
\end{equation}
and the endemic equilibrium to be

\begin{equation*}
	\begin{aligned}
		\overline{S} &= \frac{\left(b_{1}+b_{2}+b_{4}+b_{5}+b_{6}+\mu_1\right)\left(d_{1}+d_{2}+d_{4}+d_{5}+d_{6}+\mu\right)}{\alpha \beta}\\
		\overline{I} &= \frac{\alpha \beta \omega - \mu \left(b_{1}+b_{2}+b_{4}+b_{5}+b_{6}+\mu_1\right)\left(d_{1}+d_{2}+d_{4}+d_{5}+d_{6}+\mu\right)}{\alpha \beta \left(d_{1}+d_{2}+d_{4}+d_{5}+d_{6}+\mu\right)}\\
		\overline{V} &= \frac{\alpha \beta \omega - \mu \left(b_{1}+b_{2}+b_{4}+b_{5}+b_{6}+\mu_1\right)\left(d_{1}+d_{2}+d_{4}+d_{5}+d_{6}+\mu\right)}{\beta \left(b_{1}+b_{2}+b_{4}+b_{5}+b_{6}+\mu_1\right)\left(d_{1}+d_{2}+d_{4}+d_{5}+d_{6}+\mu\right)}   \label{eqn4}
	\end{aligned}
\end{equation*}

Broadly we consider two kinds of interventions for this comparative effectiveness study.
\begin{itemize}
	\item [1.] Drugs that inhibit viral replication: All the fours interventions Arbidol, Remdesivir, Interferon, Lopinavir/Ritonavir  does this job.
	So we now choose $\alpha$ to be $\alpha(1-\epsilon_1)(1-\epsilon_2)(1-\epsilon_3)(1-\epsilon_4)$, where $\epsilon_1,\epsilon_2,\epsilon_3,\epsilon_4$ are chosen based on the  efficacy of the drugs Arbidol, Remdesivir, Interferon and Lopinavir/Ritonavir respectively.
	
	\item [2.] Drugs that block virus binding to susceptible cells : Arbidol does this job.
	So we now choose $\beta$ to be $\beta(1-\gamma)$, where $\gamma$ is the efficacy of the drug Arbidol in blocking virus binding to susceptible cells.
\end{itemize}

$\mathcal{R}_{0}$ plays a crucial role in understanding the spread of infection in the individual and $\overline{V}$ determines the infectivity of virus in an individual. Taking  the two kinds of interventions  into consideration,   we now have modified   basic reproductive number $\mathcal{R}_{E}$  and modified virus count $\overline{V}_{E}$ of the endemic equilibrium to be 

\begin{equation*}
	\begin{aligned}
		\mathcal{R}_{E} &=  \frac{ \alpha(1-\epsilon_1)(1-\epsilon_2)(1-\epsilon_3)(1-\epsilon_4) \beta(1-\gamma) \omega}{\mu (b_{1}+b_{2}+b_{4}+b_{5}+b_{6}+\mu_1)(d_{1}+d_{2}+d_{4}+d_{5}+d_{6}+\mu)} \\
		\overline{V_{E}} &= \frac{\alpha(1-\epsilon_1)(1-\epsilon_2)(1-\epsilon_3)(1-\epsilon_4) \omega }{ \left(b_{1}+b_{2}+b_{4}+b_{5}+b_{6}+\mu_1\right)\left(d_{1}+d_{2}+d_{4}+d_{5}+d_{6}+\mu\right)} -  \frac{ \mu}{\beta(1-\gamma)}   \label{eqn5}
	\end{aligned}
\end{equation*}

The efficacy of these interventions is taken based on hazard ratios.
For Arbidol, we choose $\gamma =$ 0.7 and $\epsilon_1 = $ 0.173 ; 
for Remdesivir, we choose $\epsilon_2$ = 0.67 ;
for Interferon, we choose $\epsilon_3$ = 0.625;
for Lopinavir/Ritonavir, we choose $\epsilon_4$ = 0.6; \\

We now do the comparative effectiveness study of these interventions by calculating the percentage reduction of $\mathcal{R}_0$ and $\overline{V}$ for single and multiple combination of these  interventions. Percentage reduction of $\mathcal{R}_0$ and  $\overline{V}$ are given by

Percentage reduction of $\mathcal{R}_{0}  = \bigg[ \frac{\mathcal{R}_{0} - \mathcal{R}_{E_{j}}}{\mathcal{R}_{0}} \bigg] \times 100$

Percentage reduction of $\overline{V} = \bigg[ \frac{\overline{V} - \overline{V}_{E_j}}{\overline{V}} \bigg] \times 100$

where j stands for $\epsilon_1,\epsilon_2,\epsilon_3,\epsilon_4,\gamma$ or combinations thereof.

Since we have 4 drugs, we consider 16 ($= 2^4$) different combinations of these drugs.

\begin{table}[!htb]
	\centering
	\begin{tabular}{| c | c | c | c | c | c |}
		\hline
		& & \textbf{\%age change} & & \textbf{\%age change} & \\
		\textbf{No.} & \textbf{Intervention} & \textbf{in $\mathcal{R}_{0}$} & \textbf{Rank}  & \textbf{in $\overline{V}$ } & \textbf{Rank} \\
		\hline
		1 & Nil & 0  &  1  &  0  &  1  \\
		\hline
		2 & $\epsilon_1 \gamma$ & 75.19  &  5  &  270.39  &  9  \\
		\hline
		3 & $\epsilon_2$ & 67  &  4  &  134.66  &  4  \\
		\hline
		4 & $\epsilon_3$ & 62.5	 & 	3	 & 	125.61	 & 	3	 \\
		\hline
		5 & $\epsilon_4$ & 60	 & 	2	 & 	120.59	 & 	2	 \\
		\hline
		6 & $\epsilon_2 \epsilon_1 \gamma$ & 91.81	 & 	11	 & 	381.75	 & 	12	 \\
		\hline
		7 & $\epsilon_3 \epsilon_1 \gamma$ & 90.7	 & 	10	 & 	374.27	 & 	11	 \\
		\hline
		8 & $\epsilon_4 \epsilon_1 \gamma$ & 90.08	 & 	9	 & 	370.12	 & 	10	 \\
		\hline
		9 & $\epsilon_2 \epsilon_3$ & 87.63	 & 	8	 & 	176.11	 & 	7	 \\
		\hline
		10 & $\epsilon_2 \epsilon_4$ & 86.8	 & 	7	 & 	174.45	 & 	6	 \\
		\hline
		11 & $\epsilon_3 \epsilon_4$ & 85	 & 	6	 & 	170.83	 & 	5	 \\
		\hline
		12 & $\epsilon_2 \epsilon_4 \epsilon_1 \gamma$ & 96.73	 & 	14	 & 	414.66	 & 	14	 \\
		\hline
		13 & $\epsilon_2 \epsilon_3 \epsilon_1 \gamma$ & 96.93	 & 	15	 & 	416.03	 & 	15	 \\
		\hline
		14 & $\epsilon_3 \epsilon_4 \epsilon_1 \gamma$ & 96.28	 & 	13	 & 	411.67	 & 	13	 \\
		\hline
		15 & $\epsilon_2 \epsilon_3 \epsilon_4$ & 95.05	 & 	12	 & 	191.03	 & 	8	 \\
		\hline
		16 & $\epsilon_1 \epsilon_2 \epsilon_3 \epsilon_4 \gamma$ & 98.77	 & 	16	 & 	428.37	 & 	16	 \\
		\hline
	\end{tabular} 
	\caption{Comparative Effectiveness Study} \label{comp}
\end{table} \vspace{.25cm}

In the table \ref{comp} the comparative effectiveness is calculated and measured on a scale from 1 to 16 with 1 denoting the lowest comparative effectiveness while 16 denoting the highest comparative effectiveness. The conclusions from this study are the following.
\begin{itemize}
	\item [1.] When single drug/intervention is administered, Arbidol outperforms other drugs/interventions w.r.t reduction  both $\mathcal{R}_{0}$ and $\overline{V}$ (refer rows 2 to 5 in table \ref{comp}). 
	\item [2.] When two drugs/interventions are administered, Remdesivir and Arbidol combination performs better than any combination of two drugs/interventions in reducing $\mathcal{R}_{0}$ and $\overline{V}$ (refer rows 6 to 11 in table \ref{comp}).
	\item [3.] When three drugs/interventions are administered, Remdesivir,Interferon and Arbidol combination performs better than any combination of three drugs/interventions in reducing $\mathcal{R}_{0}$ and $\overline{V}$ (refer rows 12 to 15 in table \ref{comp}).
	\item [4.] The best reduction in $\mathcal{R}_{0}$ and $\overline{V}$ is seen (refer row 16 in table \ref{comp}) when all the four drugs/interventions are applied in combination.
\end{itemize}

\section{ Discussion and Conclusions} \vspace{.25cm}

\quad 

\quad In this work we have considered four drug interventions namely, Arbidol, Remdesivir, Interferon, Lopinavir/Ritonavir  and studied their efficacy for treatment of COVID-19 when applied individually or in combination. This study is done in two ways. \vspace{.25cm}

\quad The first by modeling these interventions as control interventions and studying the optimal control problem. In this approach, we find a control for a Dynamical system over a period of time such that an objective function is optimised. It has numerous applications in both science and engineering. For this modelling study, the dynamical system is the dynamics of cells and virus in a human body. An optimal control problem with drug-treatment control on the within host dynamical system is proposed and studied.
The second study was  the comparative effectiveness study. This research study adopts many of the same approaches and methodologies as cost-effectiveness analysis, including the use of incremental cost-effectiveness ratios (ICERs) and quality-adjusted life years (QALYs). Engaging various strategies in this process, while difficult, makes research more applicable through providing information that improves patient decision making \vspace{.25cm}

\quad  Conclusions  from the optimal control studies and comparative effectiveness studies suggest  the following.
\begin{itemize}
	\item [1.]  All the drugs when administered individually or in combination reduce the infected cells and viral load significantly.
	\item [2.]  The  average  infected cell count and  viral load decreased the most when all the four interventions were applied together.
	\item [3.]  The average susceptible cell count decreased the best when  Arbidol alone was administered.
	
	\item [4.] The best reduction in basic reproduction number and viral count was obtained when all the four drugs/interventions are applied in combination.
\end{itemize}

\vspace{.25cm}
{\flushleft{  \textbf{ACKNOWLEDGEMENTS} }}\vspace{.25cm}

The authors from SSSIHL  acknowledge the support of SSSIHL administration for this work.  \vspace{.25cm}
	
	\bibliographystyle{amsplain}
	\bibliography{reference}
	
	\section{Authors Contribution }  \vspace{.25cm} 
	
	\begin{itemize}
		\item{BC carried out data acquisition, data analysis.}
		\item{VMB conceptualized the study, framed objectives and study design, involved in data interpretation and reviewed the manuscript. }
		\item{DKKV contributed substantial in conducting mathematical modelling of the data and prepared the manuscript.}
		\item{AVS conducted data analysis and application of mathematical functions.}
		\item{BP played crucial role in data acquisition, analysis and interpretation of key mathematical functions. }
		\item{SM helped in conceptualization, data synthesis for biochemical parameters, manuscript preparation.}
		\item{PD conducted critical review and synthesis of the manuscript.}
		\item{CBS helped in critical review of the manuscript and provided suggestions for improvement.}
	\end{itemize}
	
	\vspace{.25cm}

\end{document}